\documentclass[English]{article}       
\usepackage{graphicx}
\usepackage[ruled,vlined,algo2e,linesnumbered]{algorithm2e}
\usepackage{amsmath}
\usepackage{amssymb}
\usepackage{fullpage}

\newtheorem{theorem}{Theorem}

\newtheorem{lemma}[theorem]{Lemma}
\newtheorem{corollary}[theorem]{Corollary}

\newenvironment{proof}{\noindent {\bf Proof }}{\qed}

\newcommand{\qed}{\penalty 1000 \hfill\penalty 1000$\Box$\par\medskip}
\newcommand{\m}[1]{B\!\left(#1\right)}
\begin{document}

\title{The complexity of solving reachability games using value and
  strategy iteration \thanks{Work supported by Center for
    Algorithmic Game Theory, funded by the Carlsberg Foundation.
The authors acknowledge support from The Danish National Research
  Foundation and The National Science Foundation of China (under the
  grant 61061130540) for the Sino-Danish Center for the Theory of
  Interactive Computation, under which part of this work was performed. A preliminary version of this paper appeared in the proceedings of CSR'11.}}

\author{Kristoffer Arnsfelt Hansen  \\       
Department of Computer Science \\
           Aarhus University
\and
        Rasmus Ibsen-Jensen \\
Department of Computer Science \\
           Aarhus University
\and
        Peter Bro Miltersen \\
Department of Computer Science \\
           Aarhus University
}

\maketitle

\begin{abstract}
Two standard algorithms for approximately solving
  two-player zero-sum concurrent reachability games are {\em value
    iteration} and {\em strategy iteration}.  
We prove upper and lower bounds of $2^{m^{\Theta(N)}}$ 
on the worst case number of
  iterations needed by both of these algorithms for providing
  non-trivial approximations to the value of a game with $N$
  non-terminal positions and $m$ actions for each player in each
  position.  
 In particular, both algorithms have doubly-exponential complexity. 
 Even when the game given as
  input has only one non-terminal position, we prove an exponential
  lower bound on the worst case number of iterations needed to provide non-trivial approximations.
\end{abstract}

\section{Introduction}
\subsection{Statement of problem and overview of results}
We consider {\em finite state, two-player, zero-sum, deterministic,
  concurrent reachability games}. For brevity, we shall henceforth
refer to these as just reachability games. The class of reachability
games is a subclass of the class of games dubbed {\em recursive games}
by Everett \cite{Everett} and was introduced to the computer science
community in a seminal paper by de Alfaro, Henzinger and Kupferman
\cite{AHK}.  A reachability game $G$ is played between two players,
Player I and Player II. The game has a finite set of {\em
  non-terminal} positions and special {\em terminal} positions GOAL
and TRAP.
\footnote{Including the TRAP position in the setup is
  actually not strictly needed, as one could replace it with any
  non-terminal position from which no escape is possible, but
  including it is quite convenient and fairly standard. In particular,
  including it makes ``a reachability game with {\em one} non-terminal  
 position'' mean what we think it should.} 
In this paper, we let $N$
denote the number of non-terminal positions and assume positions are
indexed $1,\ldots,N$ while GOAL is indexed $N+1$ and TRAP $0$.  At any
point in time during play, a {\em pebble} rests at some position. The
position holding the pebble is called the {\em current} position.  The
objective for Player I is to eventually make the current position
GOAL. If this happens, play ends and Player I wins. The objective for
Player II is to {\em forever} prevent this from happening. This may be
accomplished either by the pebble reaching TRAP from where it cannot
escape or by it moving between non-terminal positions indefinitely.
To each non-terminal position $i$ is associated a finite set of
actions $A^1_i, A^2_i$ for each of the two players. In this paper, we
assume that all these sets have the same size $m$ (if not, we may
``copy'' actions to make this so) and that $A^1_i = A^2_i =
\{1,\ldots,m\}$.  At each point in time, if the current position is
$i$, Player I and Player II simultaneously choose actions in
$\{1,\ldots,m\}$. For each position $i$ and each action pair $(a,a')
\in \{1,\ldots,m\}^2$ is associated a position $\pi(i,a,a')$. In other
words, each position holds an $m \times m$ matrix of {\em pointers} to
positions.  When the current position at time $t$ is $i$ and the
players play the action pair $(a, a')$, the new position of the pebble
at time $t+1$ is $\pi(i,a,a')$.


A {\em strategy} for a reachability game is a (possibly randomized)
procedure for selecting which action to take, given the history of the
play so far. A strategy {\em profile} is a pair of strategies, one for
each player. A {\em stationary strategy} is the special case of a
strategy where the choice only depends on the current position. Such a
strategy is given by a family of probability distributions on actions,
one distribution for each position, with the probability of an action
according to such a distribution being called a {\em behavior
  probability}.  We let $\mu_i(x, y)$ denote the probability that
Player I eventually reaches GOAL if the players play using the
strategy profile $(x,y)$ and the pebble starts in position $i$.  The
{\em lower value} of position $i$ is defined as: $\underline{v_i} =
\sup_{x \in S^1} \inf_{y \in S^2} \mu_i(x, y)$ where $S^1$ $(S^2)$ is
the set of strategies for Player I (Player II).  Similarly, the {\em
  upper value} of a position $i$ is $\overline{v_i} = \inf_{y \in S^2}
\sup_{x \in S^1} \mu_i(x, y).$ Everett \cite{Everett} showed that for
all positions $i$ in a reachability game, the lower value
$\underline{v_i}$ in fact equals the upper value $\overline v_i$, and
this number is therefore simply called the {\em value} $v_i$ of that
position. The vector $v$ is called the {\em value vector} of the
game. Furthermore, Everett showed that for any $\epsilon > 0$, there
is a stationary strategy $x^*$ of Player I so that for all positions
$i$, we have $\inf_{y \in S^2} \mu_i(x^*, y) \geq v_i - \epsilon,$
i.e. the strategy $x^*$ guarantees the value of any position within
$\epsilon$ when play starts in that position.  Such a strategy is
called {\em $\epsilon$-optimal}.  Note that $x^*$ does not depend on
$i$. It may however depend on $\epsilon > 0$ and this dependence may
be necessary, as shown by examples of Everett. In contrast, it is
known that Player II has an exact optimal strategy that is guaranteed
to achieve the value of the game, without any additive error
\cite{BAMS:Parthasarathy71,PAMS:HPRV76}.

In this paper, we consider algorithms for {\em solving} reachability
games. There are two notions of solving a reachability game relevant
for this paper:
\begin{enumerate}
\item{}{\em Quantitatively}: Given a game, compute 
$\epsilon$-approximations of the entries of its
value vector (we consider approximations, rather than
    exact computations, as the value of a reachability game may be an
    irrational number).
\item{}{\em Strategically}: Given a game, 
compute an $\epsilon$-optimal strategy for Player I.
\end{enumerate}
Once a game has been solved strategically, it is straightforward to
also solve it quantitatively (for the same $\epsilon$) by analyzing,
using linear programming, the finite state Markov decision process for
Player II resulting when freezing the computed strategy for Player I.
The converse direction is far from obvious, and it was in fact shown
by Hansen, Kouck\'{y} and Miltersen \cite{purgatory} that if standard
binary representation of behavior probabilities is used, merely {\em
  exhibiting} an $(1/4)$-optimal strategy requires worst case
exponential space in the size of the game. In contrast, a
$(1/4)$-approximation to the value vector obviously only requires
polynomial space to describe and it may be possible to compute it in
polynomial time, though it is currently not known how to do so
\cite{OtherISAAC}.

There is a large and growing literature on solving reachability games
\cite{AHK,EY,CMJ04,ChatQest,ChatSODA09,purgatory}.
In this paper, we focus on the two perhaps best-known and best-studied
algorithms,
{\em value iteration} and {\em strategy iteration}. Both were originally
derived from similar algorithms for solving Markov decision processes
\cite{Howard} and {\em discounted} stochastic games \cite{Shapley}.
We describe these algorithms next.
Value iteration is Algorithm \ref{alg-vi}. Value iteration approximately solves
reachability games quantitatively. 

\begin{center}
\begin{minipage}{0.85\linewidth}
\begin{algorithm2e}[H]
\caption{Value Iteration\label{alg-vi}}
$t:=0$ \;
$\tilde v^0 := (0,\ldots,0,1)$ \tcp{the vector $\tilde v^0$ is indexed $0,1,\ldots,N, N+1$}
\While{{\bf true}}{
 $t := t+1$ \;
$\tilde v^t_0 := 0$ \;
$\tilde v^t_{N+1} := 1$ \;
\For{$i \in \{1,2,\ldots,N\}$}{
$\tilde v^t_i := \mbox{\rm val}(A_i(\tilde v^{t-1}))$  \;
}
}
\end{algorithm2e}
\end{minipage}
\end{center}
\begin{center}
\begin{minipage}{0.85\linewidth}
\begin{algorithm2e}[H]
\caption{Strategy Iteration \label{alg-si}}
 $t:=1$ \;
 $x^1 := $ the strategy for Player I playing uniformly at each position\;
\While{{\bf true}}{
$y^t := $ an optimal {\em best reply} by Player II to $x^t$ \;
\For{$i \in \{0,1,2,\ldots,N,N+1\}$}{
$v^t_i := \mu_i(x^t,y^t)$ \;
}
$t := t+1$\;
\For{$i \in \{1,2,\ldots,N\}$}{
\eIf{${\rm val}(A_i(v^{t-1})) > v^{t-1}_i$}{
$x^t_i := \mbox{\rm maximin}(A_i(v^{t-1}))$ \;
}{
$x^t_i := x^{t-1}_i$ \;
}}}
\end{algorithm2e} 
\end{minipage}
\end{center}

In the pseudocode of Algorithm \ref{alg-vi}, the matrix $A_i(\tilde
v^{t-1})$ denotes the result of replacing each pointer to a position
$j$ in the $m \times m$ matrix of pointers at position $i$ with the
real number $\tilde v^{t-1}_j$. That is, $A_i(\tilde v^{t-1})$ is a
matrix of $m \times m$ real numbers. Also, val$(A_i(\tilde v^{t-1}))$
denotes the value of the {\em matrix game} with matrix $A_i(\tilde
v^{t-1})$ and the row player being the maximizer. This value may be
found using linear programming.  Value iteration works by iteratively
updating a {\em valuation} of the positions, i.e., the numbers $\tilde
v^t_i$. Clearly, when implementing the algorithm, valuations $\tilde
v^t_i$ only have to be kept for one iteration of the while loop after
the iteration in which they are computed and the algorithm thus only
needs to store $O(N)$ real numbers.\footnote{
In this paper, we assume the real
  number model of computation and ignore the (severe) technical issues
  arising when implementing the algorithm using finite-precision
  arithmetic.}  As stated, the algorithm is non-terminating, but has
the property that as $t$ approaches infinity, the valuations $\tilde
v^t_i$ approach the correct values $v_i$ from below. We present an
easy (though not self-contained) proof of this well-known fact in
section \ref{sec-pre} below, and also explain the intuition behind the
truth of this statement. However, until the present paper, there
  has been no published information on the number of iterations needed
  for the approximation to be an $\epsilon$-approximation to the
  correct value for the general case of concurrent reachability
  games, though Condon \cite{Con} observed that for the case of {\em
  turn-based} games (or ``simple stochastic games''), the number of
iterations has to be at least exponential in $N$ in order to achieve
an $\epsilon$-approximation. Clearly, the concurrent case is at 
least as bad. In
fact, this paper will show that the concurrent case is in fact much worse.

Strategy iteration is Algorithm \ref{alg-si}. It approximately solves
reachability games quantitatively as well as strategically.  In the
pseudocode of Algorithm \ref{alg-si}, the line ``$y^t := $ an optimal
{\em best reply} to $x^t$'' should be interpreted as follows: When
Player I's strategy has been ``frozen'' to $x^t$, the resulting game
is a one-player game for Player II, also known as an {\em absorbing
  Markov decision process}.  For such a process, an optimal stationary
strategy $y^t$ that is pure is known to exist, and can be found in
polynomial time using linear programming \cite{Howard}.  The
expression $\mbox{\rm maximin}(A_i(v^{t-1}))$ denotes a maximin mixed
strategy (an ``optimal strategy'') for the maximizing row player in
the matrix game $A_i(v^{t-1})$.  This optimal strategy may again be
found using linear programming.  The strategy iteration algorithm was
originally described for {\em one-player games} by Howard
\cite{Howard}, with Player I being the single player -- in
that case, in the pseudocode, the line ``$y^t := $ an optimal {\em
  best reply} to $x^t$'' is simply omitted. Subsequently, a variant of
the pseudocode of Algorithm \ref{alg-si} was shown by Hoffman and Karp
\cite{HK} to be a correct approximation algorithm for the class of
{\em recurrent undiscounted stochastic games} and by Rao,
Chandrasekaran and Nair \cite{RCN} to be a correct algorithm for the
class of {\em discounted stochastic games}. Finally, Chatterjee, de
Alfaro and Henzinger \cite{ChatQest} showed the pseudocode of
Algorithm \ref{alg-si} to be a correct approximation algorithm for the
class of reachability games.
As is the case for value
iteration, the strategy iteration algorithm is non-terminating, but
has the property that as $t$ approaches infinity, the valuations
$v^t_i$ approach the correct values $v_i$ from below. Chatterjee {\em
  et al.} \cite[Lemma 8]{ChatQest} prove this by relating the
algorithm to the value iteration algorithm. In particular, they prove:
\begin{equation}\label{ineq-val-strat}
\tilde v_i^t \leq v_i^t \leq v_i.
\end{equation}
That is, strategy iteration needs at most as many iterations of the
while loop as value iteration to achieve a particular degree of
approximation to the correct values $v_i$. Also, the strategies $x^t$
{\em guarantee} the valuations $v_i^t$ for Player I, so whenever these
valuations are $\epsilon$-close to the values, the corresponding $x^t$
is an $\epsilon$-optimal strategy.  However, until the present
  paper, there has been no published information on the number of
  iterations needed for the approximation to be an $\epsilon$-optimal
  solution, though a recent breakthrough result of Friedman
\cite{Fried} proved
that for the case of {\em
  turn-based} games, the number of
iterations is at least exponential in $N$ in the worst case. Clearly,
the concurrent case is at least as bad. In fact, this paper will show
that the concurrent case is much worse!

As our main result, we exhibit a family of reachability
games with $N$ positions and $m$ actions for each player in each
position, such that all non-terminal positions have value one and
such that value iteration as well as strategy iteration need at least a
{\em doubly exponential} $2^{m^{\Omega(N)}}$ number of iterations to
obtain valuations larger than any fixed constant (say $0.01$). By
inequality (\ref{ineq-val-strat}), it is enough to consider the
strategy iteration algorithm to establish this. However, our proof is
much easier and cleaner for the value iteration algorithm, the exact
bounds are somewhat better, and our much more technical proof for the
strategy iteration case is in fact based upon it. So, we shall present
separate proofs for the two cases.

Our hard instances $P(N,m)$ for both algorithms are generalizations of
the ``Purgatory'' games defined by Hansen, Miltersen and Kouck\'{y}
\cite{purgatory} (these occur as special cases by setting $m=2$).
Following the conventions of that paper, we describe these games as
being games between {\em Dante} (Player I) and {\em Lucifer} (Player
II). The game $P(N,m)$ can be described succinctly as follows:
{\em 
  Lucifer repeatedly selects and hides a number between 1 and
  $m$. Each time Lucifer hides such a number, Dante must try to guess
  which number it is. After the guess, the hidden number is revealed.
  If Dante ever guesses a number which is strictly higher than the one
  Lucifer is hiding, Dante loses the game. If Dante ever guesses
  correctly $N$ times in a row, the game ends with Dante being the
  winner.  If neither of these two events ever happen and the play
  thus continues forever, Dante loses.
}
It is easy to see that $P(N,m)$ can be described as a deterministic
concurrent reachability game with $N$ non-terminal positions and
$m$ actions for each player in each position. Also, by applying
a polynomial-time algorithm by de Alfaro {\em et al.} \cite{AHK} for determining
which positions in a reachability game have value 1, we find that
all positions except TRAP have value 1 in $P(N,m)$. That is, Dante
can win this game with arbitrarily high probability.

We note that these hard instances are very natural and 
easy to describe as games that one might even conceivably have 
a bit of fun playing (the reader is invited to try playing
$P(3,2)$ or $P(1,5)$ with an uninitiated party)!
In this respect, they are quite different from the recent 
extremely ingenious turn-based games due to Friedman \cite{Fried} 
where strategy iteration exhibits exponential behavior.

Using recent improved upper bounds on the \emph{patience} of
$\epsilon$-optimal strategies for Everett's recursive games, we
provide matching $2^{m^{O(N)}}$ upper bounds on the
number of iterations sufficient for getting adequate approximate
values, by each of the algorithms. In particular, both algorithms are
also of {\em at most} doubly-exponential complexity. 

\begin{table}[ht]
\center
\label{table}
\begin{tabular}{llllllllll}
\hline\noalign{\smallskip}
\# Iterations & $10^0$ & $10^1$ & $10^2$ & $10^3$ & $10^4$ & $10^5$ &
$10^6$ & $10^7$ & $10^8$ \\
\noalign{\smallskip}\hline\noalign{\smallskip}
Valuation & 0.013 & 0.035 & 0.069 & 0.102 & 0.134  & 0.165 & 0.194  & 0.223  & 0.248 \\
\noalign{\smallskip}\hline
\end{tabular}
\caption{Running Strategy Iteration on $P(7,2)$.}
\end{table}

That the doubly-exponential complexity is a real phenomenon is
illustrated in Table \ref{table} which tabulates the valuations
computed by strategy iteration for the initial position of $P(7,2)$, i.e.,
``Dante's Purgatory'' \cite{purgatory}, a 7-position game of
value 1. The algorithm was implemented using double precision floating point
arithmetic and was allowed to run for one hundred
million iterations at which point the arithmetic precision was
inadequate for representing the computed strategies (note that the main result
of Hansen, Miltersen and Kouck\'{y} \cite{purgatory}
implies that roughly 64 {\em decimal}
digits of precision is needed to describe a strategy achieving 
a valuation above 0.9).

Interestingly, when introduced as an algorithm for solving concurrent
reachability games \cite{ChatQest}, strategy iteration was proposed as
a practical alternative to generic algorithms having an exponential
worst case complexity. More precisely, one obtains a generic algorithm
for solving reachability games quantitatively by reducing the problem
to the decision problem for the existential fragment of the first
order theory of the real numbers \cite{EY}. This yields an
exponential time (in fact a {\bf PSPACE}) algorithm.
Our results show that this generic algorithm is in fact astronomically
more practical than strategy iteration on very simple and natural
instances. Still, it is not practical in any real sense of this term,
even given state-of-the-art implementations of the best known decision
procedures for the theory of the reals. Finding a practical algorithm
remains a very interesting open problem.
\subsection{Overview of proof techniques}
Our proof of the lower bound for the case of value 
iteration is very intuitive. It is based on combining the following facts:
\begin{enumerate}
\item{}The valuations $\tilde v^t_i$ obtained in iteration $t$ of value iteration is in fact the values of a {\em time bounded} version of the reachability game, where Player I loses if he has not reached GOAL at time $t$.
\item{}While the value of the game $P(N,m)$ is 1, the value of its
  time bounded version is very close to 0 for all small values of $t$.
\end{enumerate}
The second fact was established by Hansen {\it et
  al.} \cite{purgatory} for the case $m=2$ by relating the so-called
{\em patience} of reachability games to the values of their time
bounded version, without the connection to the value iteration
algorithm being made explicit, by giving bounds on the patience of the
games $P(N,2)$. The present paper provides a different and arguably
simpler proof of the lower bound on the value of the time bounded game
that gives bounds also for other values of $m$ than 2. It is based on
exhibiting a fixed strategy for Lucifer that prevents Dante from
winning fast.

The lower bound for strategy iteration is much more technical.
We remark that the analysis of value iteration is
used twice and in two different ways in the proof.
It proceeds roughly as follows: The analysis of value iteration
yields that when value iteration is applied to $P(1,m)$,
exponentially many iterations (in $m$) are needed to yield 
a close approximation of the value.
We can also show that when strategy iteration is applied to $P(1,m)$,
exactly the same sequence of valuations is computed as when value
iteration is applied to the same game.
From these two facts, we can derive an upper bound on the
patience of the strategies computed by strategy
iteration on $P(1,m)$.
Next, a quite involved argument shows that when
applying strategy iteration to $P(N,m)$, the sequence of strategies computed
for one of the positions (the initial one) is exactly the
same as the one computed when strategy iteration is applied
to $P(1,m)$. We also show that the smallest behavior probability in
the computed strategy for $P(N,m)$ occurs in the initial position.
In particular, the patiences of the sequence of strategies computed for
$P(N,m)$ is the same as the patiences of the sequence of strategies
computed for $P(1,m)$.
Finally, our analysis of value iteration for $P(N,m)$ and the relationship
between patience and value iteration allow us to conclude that a
strategy with low patience for $P(N,m)$ cannot be near-optimal, 
yielding the desired doubly-exponential lower bound. 

\section{Theorems and Proofs}
\subsection{The connection between patience, the value of time bounded
  games, and the complexity of value iteration}\label{sec-pre}
The key to understanding value iteration is the following folklore lemma.
Given a concurrent reachability game $G$, we define $G_T$ to be the {\em finite}
extensive form game with the same rules as $G$, except that 
Player 1 loses if he has not reached GOAL after $T$ moves of
the pebble. The positions of $G_T$ are denoted by $(i,t)$, where
$i$ is a position of $G$ and $t$ is an integer denoting the number
of time steps left until Dante's time is out. 
\begin{lemma}\label{lem-valit}
The valuation $\tilde v^t_i$ computed by the value iteration algorithm when
applied to a game $G$ is the exact value of position
$(i,t)$ in the game $G_t$.
\end{lemma}
The proof is an easy induction in $t$ (``Backward induction'').
A very general result by Mertens and Neyman \cite{MertensNeyman} establishes
that for a much more general class of games (undiscounted stochastic
games), the value of the
time bounded version converges to the value of the infinite version
as the time bound approaches infinity.
Combining this with Lemma \ref{lem-valit} immediately yields the 
correctness of the value iteration algorithm.

The {\em patience} \cite{Everett} 
of a stationary strategy for a concurrent reachability
game is $1/p$, where $p$ is the smallest non-zero behavior probability
employed by the strategy in any position. The following lemma relates
the patience of near-optimal strategies of a reachability game
to the difference between the values of 
the time bounded and the infinite game and hence to the convergence rate of
value iteration. 
\begin{lemma}\label{lem-patience2time}
  Let $G$ be a reachability game with $N$ non-terminal positions and
  with an $\epsilon$-optimal strategy of patience at most $l$, for
  some $l \geq 1, \epsilon > 0$.  Let $T = k N l^N$ for some $k \geq
  1$, and $u$ be any position of $G$. Then, the value of position
  $(u,T)$ of $G_T$ differs from the value of the position $u$ of $G$
  by at most $\epsilon + e^{-k}$.
\end{lemma} 
\begin{proof}
We want to show that the value of $(u,T)$ in $G_T$
is at least $v_u - \epsilon - e^{-k}$, where $v_u$ is the value of position
$u$ in $G$. We can assume that $v_u>\epsilon$, because otherwise we are done.
Fix an $\epsilon$-optimal stationary strategy $x$ for Dante in 
$G$ of patience at most $l$.
Consider this as a strategy of $G_T$ and consider play starting in $u$. 
We shall show that $x$ guarantees Dante to win $G_T$ with
probability at least $v_u - \epsilon - e^{-k}$, thus proving the
statement.
Consider a best reply $y$ by Lucifer to $x$ in $G_T$.
Note that $y$ does not necessarily correspond to a stationary strategy
in $G$. The strategy can still be played by Lucifer in $G$, by playing 
by it for the first $T$ time steps and playing arbitrarily afterwards. 

Call a position $v$ of $G$ {\em alive} if there are paths from $v$ to
GOAL in {\em all} directed graphs obtained from $G$ in the following
way: The nodes of the graphs are the positions of $G$. We then select
for each position an arbitrary column for the corresponding matrix,
and let the edges going out from this node correspond to the pointers
of the chosen column and rows where Dante assigns positive
probability. That is, intuitively, a position $v$ is alive, if and
only if there is no {\em absolutely sure} way for Lucifer for
preventing Dante from reaching GOAL when play starts in $v$.
Positions that are not alive are called {\em dead}. Note that if a
position $v$ is dead, the strategy $y$, being a best reply of Lucifer,
will pick actions so that the probability of play reaching GOAL,
conditioned on play having reached $v$, is 0.  On the other hand, if
the current position $v$ is alive, the conditional probability that
play reaches GOAL within the next $N$ steps is at least $(1/l)^N$.
That is, looking at the entire play, the probability that play has
{\em not} reached either GOAL or a dead state after $T$ steps is at
most $(1-l^{-N})^{T/N} = (1-l^{-N})^{k l^N} \leq e^{-k}$.  Suppose now
that GOAL is reached in $T$ steps with probability strictly less than
$v_u - \epsilon - e^{-k}$ when play starts in $u$. This means that a
dead position is reached with probability strictly greater than
$1-(v_u - \epsilon - e^{-k}) - e^{-k}$, i.e., strictly greater than
$1-(v_u-\epsilon)$. But this means that if Lucifer plays $y$ as a reply
to $x$ in the {\em infinite} game $G$ he will in fact succeed in
getting the pebble to reach a dead position and hence prevent Dante
from {\em ever} reaching GOAL, with probability strictly greater than
$1-(v_u-\epsilon)$.  This contradicts $x$ being $\epsilon$-optimal for
Dante in $G$.  Thus, we conclude that GOAL is in fact reached in $T$
steps with probability at least $v_u - \epsilon - e^{-k}$ when play
starts in $u$ with $x$ and $y$ being played against each other in
$G_T$, as desired.
\end{proof}
The connection between the convergence of value iteration and
the time bounded version of the game allows us to reformulate the
lemma in the following very useful way.
\begin{lemma}\label{lem-patience2valit}
Let $G$ be a reachability game with an $\epsilon$-optimal
strategy of patience at most $l$, for some $\epsilon > 0$.
Then, $T = k N l^N$ rounds of value iteration is sufficient
to approximate the values of all positions of the game with additive error
at most $\epsilon + e^{-k}$.
\end{lemma} 
We can use this lemma to prove our upper bound on the number
of iterations of value iteration (and hence also strategy iteration).
The following lemma is from Hansen {\em et al.} \cite{HKLMT10}.
\begin{lemma}[Hansen, Kouck\'{y}, Lauritzen, Miltersen and Tsigaridas]
  \label{lem-HKLMT}
  Let $\epsilon > 0$ be arbitrary. Any concurrent reachability game
  with $N$ positions and at most $m \geq 2$ actions in each position has an
  $\epsilon$-optimal stationary strategy of patience at most
  $(1/\epsilon)^{m^{O(N)}}$.
\end{lemma}
This lemma is an asymptotic improvement of Theorem 4 of Hansen {\em et
  al.} \cite{purgatory}, that gave an upper bound of
$(1/\epsilon)^{2^{30 M}}$, for a \emph{total} number of $M$ actions,
when $M \geq 10$ and $0 < \epsilon < \frac{1}{2}$. This result does however
have the advantage of an \emph{explicit} constant in the exponent,
which the bound of Lemma \ref{lem-HKLMT} lacks.

Combining Lemma \ref{lem-patience2valit}, Lemma \ref{lem-HKLMT}, and
also applying inequality (\ref{ineq-val-strat}), we get the following
upper bound:
\begin{theorem}
  Let $\epsilon > 0$ be arbitrary.  When applying value iteration or
  strategy iteration to a concurrent reachability game with $N$
  non-terminal positions and $m \geq 2$ choices for each player in
  each position, after at most $(1/\epsilon)^{m^{O(N)}}$ iterations,
  an $\epsilon$-approximation to the value has been obtained.
\end{theorem}
Also, Lemma \ref{lem-patience2valit}
will be very useful for us below when applied in the contrapositive. 
Specifically, below, we
will directly analyze and compare the value of $P(N,m)$ with
the value of its time bounded version, and use this to conclude
that the value iteration algorithm does not converge quickly when
applied to this game.
The lemma then implies that the patience of any $\epsilon$-optimal
strategy is large. When we later consider the strategy iteration algorithm
applied to the same game, we will show that the strategy
computed after any sub-astronomical number of iterations has 
too low patience to be $\epsilon$-optimal.

\subsection{The value of time bounded Generalized Purgatory and the
complexity of value iteration}\label{sec-val}
In this section we give an upper bound on the value of a time bounded
version of the Generalized Purgatory game $P(N,m)$. 
As explained in Section \ref{sec-pre}, 
this upper
bound immediately implies a lower bound on the number of iterations
needed by value iteration to approximate the value of the
original game.

We let $P_T(N,m)$ be the time bounded version of
$P(N,m)$ as defined in Section \ref{sec-pre}, i.e.
$P_T(N,m)$ is syntactic sugar for $(P(N,m))_T$.
Also, we need to fix an indexing of the positions of $P(N,m)$.
We define position $i$ for $i=1,\ldots,N$ to be the position
where Dante already guessed correctly $i-1$ times in a row and
still needs to guess correctly $N-i+1$ times in a row to win the game.
 
First we give a rather precise analysis of the one-position case.
Besides being interesting in its own right (to establish that
value iteration is exponential even for this case), this will 
also be useful later when we analyze strategy iteration.
\begin{theorem}\label{thm-valitonepos}
Let $m \geq 2$ and $T \geq 1$. The value of position $(1,T)$
of $P_T(1,m)$ is less than
\[1-(1-\frac{1}{m})(\frac{1}{mT})^{1/(m-1)}.\]
\end{theorem}
\begin{proof}
Let $\epsilon = (1/mT)^{1/(m-1)}$.
Consider any strategy (not necessarily stationary) for Dante for
playing $P_T(1,m)$.
In each round of play, Dante chooses his action with a probability
distribution that may depend on previous play and time left.
We define a reply by Lucifer in a round-to-round fashion.

Fix a history of play leading to some current round and 
let $p_1, p_2, \ldots, p_m$ be the probabilities by which
Dante plays $1,2, \ldots, m$ in this current round.
There are two cases.
\begin{enumerate}
\item{}There is an $i$ so that 
$p_i < (\frac{1-\epsilon}{\epsilon}) \sum_{j \geq i+1} p_j$.
We call such a round a {\em green} round. In this case, Lucifer plays $i$.
\item{}For all $i$,
$p_i \geq (\frac{1-\epsilon}{\epsilon}) \sum_{j \geq i+1} p_j$.
We call such a round a {\em red} round. In this case, Lucifer plays $m$.
\end{enumerate}
This completes the definition of Lucifer's reply.

We now analyze the probability that Dante wins
$P_T(1,m)$ when he plays his strategy and Lucifer plays this reply.
We show this probability to be at most 
\[1-(1-\frac{1}{m})(\frac{1}{mT})^{1/(m-1)}\] and we shall be done.

Let us consider a green round. We claim that the
probability that Dante wins in this round, conditioned on the
previous history of play, and conditioned on play ending in this round,
is at most $1-\epsilon$.
Indeed, this conditional probability is given by 
\begin{eqnarray*}
\frac{p_i}{p_i + (p_{i+1} + \cdots + p_m)} & < & 
\frac{(\frac{1-\epsilon}{\epsilon})(\sum_{j \geq i+1} p_j)} 
{(\frac{1-\epsilon}{\epsilon})(\sum_{j \geq i+1} p_j) + (\sum_{j \geq i+1} p_j)} \\
& = & 
\frac{(1-\epsilon)/\epsilon}{(1-\epsilon)/\epsilon + \epsilon/\epsilon} \\
& = & 1-\epsilon.
\end{eqnarray*}

Let us next consider a red round. We claim that the probability of play
ending in this round, conditioned on the previous history of play, is
at most $\epsilon^{m-1}$. Indeed, note that this conditional probability
is exactly $p_m$, and that
\[ 1  =  \sum_{j=1}^m p_j 
 =  p_1 + \sum_{j=2}^m p_j 
 \geq  (1 + \frac{1-\epsilon}{\epsilon}) (\sum_{j=2}^m p_j) 
  =  (1 + \frac{1-\epsilon}{\epsilon}) (p_2 + \sum_{j=3}^m p_j) \] 
\[ \geq  (1 + \frac{1-\epsilon}{\epsilon})^2 (\sum_{j=3}^m p_j)
 \geq  \cdots 
 \geq  (1+ \frac{1-\epsilon}{\epsilon})^{m-1} p_m =
 (\frac{1}{\epsilon})^{m-1} p_m
\]

from which $p_m \leq \epsilon^{m-1}$.
That is, in every round of play, conditioned on previous play,
either it is the case that the probability that play ends
in this round is at most $\epsilon^{m-1}$ (for the case of a 
red round) or it is the case that conditioned on play ending,
the probability of win for Dante is at most $1-\epsilon$
(for the case of a green round).

Now let us estimate the probability of a win for Dante in the
entire game $P_T(1,m)$. Let $W$ denote the event that Dante wins.
Let $G$ be the event that play ends in a green round. Also, 
let $R$ be the event that play ends in a red round.
Then, we have 
\begin{eqnarray*}
\Pr[W] & = & \Pr[W | R] \Pr[R] + \Pr[W | G] \Pr[G] \\
& \leq & \Pr[R] + \Pr[W|G]\Pr[G] \\
& = & \Pr[R] + \Pr[W|G](1-\Pr[R]) \\
& = & \Pr[R] + \Pr[W|G] - \Pr[R]\Pr[W|G] \\
& < & (\epsilon^{m-1})T + (1-\epsilon) - (\epsilon^{m-1})T(1-\epsilon)\\
& = & 1 - \epsilon + T\epsilon^m \\
& = & 1 - (\frac{1}{mT})^{1/{(m-1)}} + T({\frac{1}{mT}})^{\frac{m}{m-1}}\\
& = & 1 - (1 - \frac{1}{m})(\frac{1}{mT})^{1/{(m-1)}}.
\end{eqnarray*}
\end{proof}

Combining Lemma \ref{lem-valit} with Theorem \ref{thm-valitonepos} we
get the result that value iteration needs exponential time, even
for one-position games.
\begin{corollary}\label{cor-valitonepos}
Let $0 < \epsilon <1$.
Applying less than $\frac{1}{em} (1/\epsilon)^{m-1}$ iterations
of the value iteration algorithm 
to $P(1,m)$ yields a valuation
at least $\epsilon$ smaller than the exact value.
\end{corollary}
Next, we analyze the $N$-position case, where we give a somewhat
coarser bound.
\begin{theorem}\label{thm-valitmanypos}
Let $N,m,k,T$ be integers with
$N \geq 2, m \geq 2, 1 \leq k \leq N-2$ and $T \leq 2^{m^{N-k}}$. 
Then, the value of $P_T(N,m)$ is at most 
$2m^{-k} + 2^{-m^{N-k-1}}$.
\end{theorem}
\begin{proof}
We show an upper bound on the value of $P_T(N,m)$ of 
$2m^{-k} + 2^{-m^{N-k-1}}$
by exhibiting a particular strategy of Lucifer and showing that any response
by Dante to this particular strategy of Lucifer will make 
Dante win with probability at most 
$2m^{-k} + 2^{-m^{N-k-1}}$.

To structure the proof, 
we divide the play into {\em epochs}. An epoch begins and another ends
immediately after
each time Dante has guessed incorrectly by undershooting, so that he
now finds himself in exactly the same situation as when the play
begins (but in general with less time left to win). That is,
Dante wins if and only if there is an epoch of length $N$
containing only correct guesses.
For convenience, we make the game a little more attractive
for Dante by continuing play for $T$ epochs, rather than $T$ rounds.
Call this prolonged game $G'_T$.
Clearly, the value of $G_T$ is at most the value of $G'_{T}$, so
it is okay to prove the upper bound for the latter.
We index the epochs $1,2,\ldots,T$.

To define the strategy of Lucifer, we first define a function
$f: {\bf N} \times {\bf N} \rightarrow {\bf N}$ as follows:
\[ f(i,j) = 1 + (j-1) \sum_{r=0}^{i-1} m^r.\]
Then, it is easy to see that $f$ satisfies the following two equations.
\begin{equation}\label{powereq}
f(i,m) = m^i
\end{equation}
\begin{equation}\label{sumeq}
f(i,j+1) = f(i,j) + \sum_{r=0}^{i-1}f(r,m)
\end{equation}

The specific strategy of Lucifer is this: 
Let $d$ be the number of rounds already played in the current epoch.
If $d \geq N-k$, Lucifer chooses a number
between $1$ and $m$ uniformly at random. 
If $d < N-k$, he hides the numbers $j=1,\ldots,m-1$
with probabilities $p_j(d) = 2^{-f(N-k-d, m+1-j)}$ and puts all 
remaining probability mass on the number $m$ (since $N-k-d \geq 1$ and
$m \geq 2$, there is indeed some probability mass left for $m$).

Freeze the strategy of Lucifer to this strategy. From the point of
view of Dante, the game $G_T$ is now a finite horizon absorbing Markov
decision process. Thus, he has an optimal policy that is deterministic
and history independent. That is, the choices of Dante according to
this
policy
depend only on the number of rounds already played in the present
epoch and the remaining number of epochs before the limit of $T$
epochs has been played, or, equivalently, on the index of the current
epoch. 
We can assume
without loss of generality that Dante plays such an optimal policy.
That is, his optimal policy for epoch $t$ can
be described by a specific sequence of 
actions $a_{t0}, a_{t1}, a_{t2}, \ldots, a_{t(N-1)}$ in $\{1,\ldots,m\}$ to
make in the next $N$ rounds (with the caveat that this sequence of
choices will be aborted if the epoch ends). 

Se define the following mutually exclusive
events $W_t, L_t$:
\begin{itemize}
\item{}$W_t$: Dante wins the game in  epoch $t$ (by guessing correctly $N$ times).
\item{}$L_t$: Dante loses the game in epoch $t$ (by overshooting Lucifer's number)
\end{itemize}
We make the following claim:

\noindent 
{\bf Claim:}
For each $t$, either $\Pr[W_t] \leq 2^{-m^{N-k}-m^{N-k-1}}$ 
or $\Pr[W_t]/\Pr[L_t] \leq 2m^{-k}$.

First, let us see that the claim implies the lemma. Indeed,
the probability of Dante winning can be split into the contributions
from those epochs where Dante wins with probability at most 
$2^{-m^{N-k}-m^{N-k-1}}$ 
and the remaining
epochs. The total winning probability mass from the first is
at most $T2^{-m^{N-k}-m^{N-k-1}} \leq 2^{-m^{N-k-1}}$ 
and the total winning probability
mass of the rest is at most $2m^{-k}$, giving an upper bound for
Dante's winning probability of $2m^{-k} + 2^{-m^{N-k-1}}$.

So let us prove the claim. Fix an epoch $t$ and
let $a_{t0}, a_{t1}, a_{t3}, \ldots, a_{t(N-1)}$ be Dante's sequence of actions.
Suppose $a_{t0} = 1$ and $a_{t1}=1$. 
Then, since Lucifer only plays $1$ in the first two rounds with probability
$p_1(0)p_1(1) = 2^{-f(N-k, m)}\cdot  2^{-f(N-k-1,m)}$, 
Dante only wins the game in this epoch with
at most that probability, which by equation \eqref{powereq} is
equal to $2^{-m^{N-k}-m^{N-k-1}}$, as desired.

Now assume $a_{t0} > 1$ or $a_{t1} > 1$. We want to show that 
$\Pr[W_t]/\Pr[L_t] \leq 2m^{-k}$.
Let $d$ be the largest index so that $d < N-k$ and so that
$a_{td} > 1$. Since $a_{t0} > 1$ or $a_{t1} > 1$, such a $d$ exists.
Let $E$ be the event that epoch $t$ lasts for at least $d$ rounds.
We will show that $\Pr[W_t|E]/\Pr[L_t|E] \leq 2m^{-k}$.
Since $W_t \subseteq E$, this also implies that 
$\Pr[W_t]/\Pr[L_t] \leq 2m^{-k}$.
Since we condition on $E$ we look at Dante's decision after
$d$ rounds of epoch $t$. He chooses the action $j=a_{td}>1$.
If Lucifer at this point chooses a number small than $j$,
Dante loses. 
In particular, since Lucifer chooses the number $j-1$ with probability
$2^{-f(N-k-d, m+1-(j-1)}$, Dante loses the entire game by his action $a_{td}$ 
with probability at least  $2^{-f(N-k-d, m-j)}$, conditioned on $E$.
On the other hand the probability that he wins the game in this epoch
conditioned on 
$E$ is at most
$(2^{-f(N-k-d, m+1-j)})(\prod_{i = d+1}^{N-k-1} 2^{-f(N-k-i, m)}) (m^{-k}))$,
the first factor being the probability that Lucifer chooses $j$ at
round $d$, the second factor being the probability that Lucifer 
like Dante repeatedly chooses $1$ until the last $k$ rounds of the epoch begin,
and the third factor being the probability that Lucifer matches
Dante's choices in those $k$ rounds.
Now we have
\begin{eqnarray*}
\Pr[W_t]/\Pr[L_t] & \leq & \\
\Pr[W_t|E]/\Pr[L_t|E] & \leq & \\
(2^{-f(N-k-d, m+1-j)})(\prod_{i = d+1}^{N-k-1} 2^{-f(N-k-i, m)}) (m^{-k})) 2^{f(N-k-d, m-j)} & \leq & \\
m^{-k} 2^{f(N-k-d, m-j) - f(N-k-d, m+1-j) - \sum_{r = 1}^{N-k-d-1} f(r, m)} & = & \\
2 m^{-k} 2^{f(N-k-d, m-j) - f(N-k-d, m+1-j) - \sum_{r = 0}^{N-k-d-1} f(r, m)} & = & \\
2 m^{-k}
\end{eqnarray*}
as desired.
\end{proof}
Combining Lemma \ref{lem-valit} with Theorem \ref{thm-valitmanypos} we
get the result that value iteration needs doubly exponential time to
obtain any non-trivial approximation:
\begin{corollary}
Let $N$ be even.
Applying less than $2^{m^{N/2}}$ iterations
of the value iteration algorithm 
to $P(N,m)$ yields a valuation of the initial position of
at most $3 m^{-N/2}$, even though the actual value of the game is $1$.
\end{corollary}
We also get the following bound on the patience of near-optimal
strategies of $P(N,m)$ that will be useful when analyzing strategy
iteration.
\begin{theorem}\label{thm-patpurg}
Suppose $N$ is sufficiently large and $m \geq 2$.
Let $\epsilon = 1-4m^{-N/2}$. Then all $\epsilon$-optimal strategies 
of $P(N,m)$ have patience at least $2^{m^{N/3}}$.
\end{theorem}
\begin{proof}
Putting $c=\frac{N\ln m}{2}$,
Lemma \ref{lem-patience2time} tells us
that if $P(N,m)$ has an $\epsilon$-optimal strategy of 
patience less than $l = 2^{m^{N/3}}$, then the value of 
$P_t(N,m)$ is at least $1- \epsilon - e^{-c} = 3m^{-N/2}$,
where $t = c N l^N \leq 2^{m^{N/2}}$.
But putting $k = N/2$,
Theorem \ref{thm-valitmanypos} 
tells us that the value of 
$P_t(N,m)$ is at most 
$2m^{-N/2} + 2^{-m^{N/2-1}} < 3m^{-N/2}$, a contradiction.
\end{proof}

\subsection{Strategy Iteration}
The technical content of this section is a number of lemmas on
what happens when the strategy iteration
algorithm is applied to $P(N,m)$, leading up to the following crucial lemma:
\begin{lemma}\label{lem-patsi}
When applying strategy iteration to $P(N,m)$, the 
patience of the strategy $x^t$ computed in iteration $t$ is 
at most $e \cdot m \cdot t$.
\end{lemma}
Before we prove Lemma \ref{lem-patsi}, we show that it implies
the lower bound we are looking for.
\begin{theorem}
Suppose $N$ is sufficiently large.
Applying less than $2^{m^{N/4}}$ iterations of strategy iteration
to $P(N,m)$ yields a valuation of the initial position of
less than $4 m^{-N/2}$, despite the fact that the value of the 
position is $1$.
\end{theorem}
\begin{proof}
Lemma \ref{lem-patsi} implies that
the patience of the strategy $x^t$ computed in iteration $t$ 
for $t=2^{m^{N/4}}$ is at most $e m 2^{m^{N/4}}$.
Theorem \ref{thm-patpurg} states that
if $\epsilon = 1-4m^{-N/2}$, then all $\epsilon$-optimal strategies 
of $P(N,m)$ have patience at least $2^{m^{N/3}}$. So $x^t$ is
not $\epsilon$-optimal and the bound follows.
\end{proof}

To prove Lemma \ref{lem-patsi} we need to understand strategy iteration
on $P(m,N)$ and shall through a number of lemmas establish:
\begin{itemize}
\item{}For the one-position case $P(1,m)$, value iteration and strategy iteration are ``in synch'', i.e., $\tilde v^t_i = v^t_i$ for all $i$ and $t$.
\item{}When applying strategy iteration to $P(N,m)$, the strategy computed for position $1$ after $t$ iterations is the same as that computed by 
strategy iteration applied to $P(1,m)$ after $t$ iterations.
\item{}When applying strategy iteration to $P(N,m)$, 
the smallest behavior probability computed occurs at position 1 and the
patience of the strategy computed can therefore be determined by looking
at that position.
\end{itemize}

In all lemmas below, unless otherwise mentioned, we consider 
applying the strategy iteration algorithm to $P(N,m)$
and the quantities $v^t, x^t$, etc., are those computed by
this algorithm.
\begin{lemma}
$\forall t,i \in {1,2,\ldots, N+1}:v_{i}^{t}>0$
\label{lem:positive values}
\end{lemma}
\begin{proof}
For $t=1$, we have that 
$x^1$ is the uniform distribution at each position.
We then see that $v_{n}^{1}=\frac{1}{m}>0$, since no matter
which number Lucifer chooses, Dante selects the right one with probability
$\frac{1}{m}$. We also
see that Dante has a probability of winning $i$ times in a row
of $\frac{1}{m^{i}}>0$. We therefore have that $v_{N-i+1}^{1}\geq\frac{1}{m^{i}}>0$.

We know that $v^{t+1} \geq v^t$ 
(see, e.g., Chatterjee {\em et al.} \cite{ChatQest}), 
so $\forall t,i:0<v_{i}^{0}\leq v_{i}^{t}$.\end{proof}
\begin{lemma}
$\forall t,i \in \{1,\ldots,N\},j \in \{1,\ldots,m\}:0 < x_{i,j}^{t} < 1$ 
\label{lem:positive probabillities}
\end{lemma}
\begin{proof}
Since $\forall i,t:\sum_{j=1}^{m}x_{i,j}^{t}=1$ we only need to show
that $x_{i,j}^{t}>0$.
We will do the proof by contradiction. 
Assume that $\exists t,i,j:x_{i,j}^{t}=0$.
If Lucifer replies to $x^t$ by choosing $j$ in position $i$, play reaches
GOAL with probability $0$.
Therefore $v_{i}^{t}=0$ which we showed was not the case in Lemma 
\ref{lem:positive values}.\end{proof}

\begin{lemma}
$\forall t,i:v_{i}^{t}<1$\label{lem:none 1 values}\end{lemma}
\begin{proof}
Since $\forall t,i,j:x_{i,j}^{t}>0$, by Lemma \ref{lem:positive probabillities},
we have that all strategies for Lucifer in position $i$, $y$, except
for Lucifer always choosing $m$, will make Dante lose with positive
probability. In particular, the best reply by Lucifer to $x^t$
must have that property.\end{proof}
\begin{lemma}
$\forall t,i,n:v_{i}^{t}>v_{i-1}^{t}$\label{lem:increasing values}\end{lemma}
\begin{proof}
Recall that $v_i^t$ is the winning probability of Dante if play starts in
position $i$ when he plays using $x^t$ and Lucifer plays a best reply.
By construction of $P(N,m)$ we have that any winning play starting in position
$i-1$ must subsequently visit position $i$. 
Therefore, $v_{i}^{t}\geq v_{i-1}^{t}$.
By Lemma \ref{lem:positive probabillities}
we have that, Lucifer can play 1 in position $i-1$ and hence prevent,
with positive probability, Dante from proceeding to position $i$ from position $i-1$.
Dante therefore might lose the game in position $i-1$ with positive probability. Therefore,
$v_{i}^{t}>v_{i-1}^{t}$.\end{proof}

To proceed, we need to consider the matrix games that arises when
strategy iteration is executed on $P(N,m)$. Fortunately, these are
all of a special form that can be easily analyzed.

For a real number $z$ with $0 \leq z < 1$, let $\m z$ be the $m \times m$
matrix of real numbers with $1$ in the diagonal, $0$ in all entries
below the diagonal, and $z$ in all entries above the diagonal.
Also, considering $\m z$ as a matrix game with the row player 
being the maximizer,
let $p^z$ be an optimal strategy for the row player and $q^z$ be an
optimal strategy for the column player. Finally, we let $v^z$ be
the value of the matrix game.
Straightforward calculations, which we will omit, yield the following
facts about the matrix game $\m z$.
\begin{lemma}\label{fact-useful1}
For all values $0 \leq z < 1$, the matrix game $\m z$ has the
following properties.
\begin{itemize}
\item{}
The row player has a uniquely determined optimal strategy $p^z$.
This strategy is fully mixed.
\item{}$val(\m{z}) = p^z_1 = \frac{1}{\sum_{i=0}^{m-1}(1-z)^{i}}$,
\item{}For all $i > 1$, we have that $p^z_i = 
p^z_1 (1-z)^{i-1}$.
\end{itemize}
\end{lemma}
\begin{lemma}\label{fact-useful2}
If $0 < y < z < 1$, the optimal strategies $p^y, p^z$ satisfy:
$p^y_1 < p^z_1$ and
$p^y_m > p^z_m$.
\end{lemma} 
\noindent
The connection between strategy iteration and 
the matrix game $\m{z}$ is given by:
\begin{lemma}
For all $t,i$, let $z=\left\{ \begin{array}{ccc}
0 & for & t=1\\
\frac{v_{1}^{t-1}}{v_{i+1}^{t-1}} & for & t>1\end{array}\right.$.
Under the assumption that $\forall i,t'\leq t:val\left(A_{i}\left(v^{t'}\right)\right)>v_{i}^{t'}$, the strategy $x_i^t$ computed
by strategy iteration on $P(N,m)$ is $p^z$.
\label{lem:matrix(x) corresponds to the matrix in Hoffman-Karp, assumption:update}
\end{lemma}
\begin{proof}
For $t=1$, we see that the optimal strategy for both players in
the matrix game $\m z$ which is in this case the matrix defined by the
identity matrix is to play uniformly
in $\m 0$ 
which is the same strategy as $x_{i}^{0}$ and $y_{i}^{0}$.

For $t>1$, we see that, if we update $x_{i}^{t}$, which we do by assumption, 
$x_{i}^{t}$ is the optimal solution for the row
player in the matrix game given by the $m \times m$ with
$v_{i+1}^{t-1}$ in the diagonal, 0 in all entries below the diagonal
and $v_{1}^{t-1}$ in all entries above the diagonal.
We can divide each entry in this matrix by $v_{i+1}^{t-1}$, 
per Lemma \ref{lem:positive values}. This yields the matrix
$\m{\frac{v_1^{t-1}}{v_{i+1}^{t-1}}}$.
The new matrix will have the same optimal strategies for the row player.
By Lemma \ref{lem:positive values} and \ref{lem:increasing values}
we have that $0<\frac{v_{1}^{t-1}}{v_{i+1}^{t-1}}<1$. Therefore,
$x_{i}^{t}$ is exactly $p^z$.
\end{proof}

\begin{lemma}\label{new-lemma}
When applying strategy iteration to $P(N,m)$, if
Lucifer's best reply
$y^t$ is equal to the strategy that chooses 
$1$ in all positions,
then 
\[
\frac{v_{1}^{t}}{v_{i+1}^{t}}=\prod_{j=1}^{i}x_{j,1}^{t-1}
\]
\end{lemma}
\begin{proof}
$v_{k}^{t}=\prod_{j=k}^{N}x_{j,1}^{t-1}$ from which the
statement follows.
\end{proof}

\begin{lemma}
If Lucifer's best reply $y^s$ is equal to the strategy that plays $1$
in all positions for all $s \leq t$, then $\forall i,t'\leq
t:\mbox{\rm val}(A_{i}(v^{t'}))>v_{i}^{t'}$.
\label{lem:strategy update}
\end{lemma}
\begin{proof}
We will show the statement using induction in $t'$.

We see that $\mbox{\rm val}(A_{i}(v^{t'}))=v_{i+1}^{t'}\cdot \mbox{\rm val}(A_{i}(\frac{v^{t'}}{v_{i+1}^{t'}}))=v_{i+1}^{t'}\cdot \mbox{\rm val}(B(\frac{v_{1}^{t'}}{v_{i+1}^{t'}}))$.

We can also see that $v_{i}^{t'}=v_{i+1}^{t'}\cdot x_{i,1}^{t'}$,
since we know that Lucifer played $1$ at time $t'$ (so Dante loses
if he plays $p>1$ and must win from position $i+1$ otherwise).

So we just need to show that $x_{i,1}^{t'}<\mbox{\rm val}(B(\frac{v_{1}^{t'}}{v_{i+1}^{t'}}))$.

For $t'=1$:

We can see that $x_{i,1}^{1}=\mbox{\rm val}(B(0))$.

By Lemma \ref{fact-useful1}, we have that $x_{i,1}^{t'}=\mbox{\rm val}(B(0))<\mbox{\rm val}(B(\frac{v_{1}^{1}}{v_{i+1}^{1}}))$
and the result follows.

For $t'>1$:

Since Lucifer played $1$ at time $t'-1$, we can use Lemma \ref{fact-useful1} and
Lemma \ref{lem:matrix(x) corresponds to the matrix in Hoffman-Karp, assumption:update} and get that, for all $j$, $x_{j,1}^{t'}=\mbox{\rm val}(B(\frac{v_{1}^{t'-1}}{v_{j+1}^{t'-1}}))$,
especially for $i=j$. By Lemma \ref{fact-useful2}, we just need to show that $\frac{v_{1}^{t'}}{v_{i+1}^{t'}}>\frac{v_{1}^{t'-1}}{v_{i+1}^{t'-1}}$.

We can use Lemma \ref{new-lemma} and we get that $\frac{v_{1}^{t'}}{v_{i+1}^{t'}}=\prod_{j=1}^{i}x_{j,1}^{t'}$
and that $\frac{v_{1}^{t'-1}}{v_{i+1}^{t'-1}}=\prod_{j=1}^{i}x_{j,1}^{t'-1}$.
We will show that $x_{j,1}^{t'}>x_{j,1}^{t'-1}$ and the result follows,
since $x_{j,1}^{t'-1}>0$, by Lemma \ref{lem:positive probabillities}. But since $x_{j,1}^{t'}=\mbox{\rm val}(B(\frac{v_{1}^{t'-1}}{v_{j+1}^{t'-1}}))$,
this is the induction hypothesis.
\end{proof}

\begin{lemma}
For all $t,i$, let $z=\left\{ \begin{array}{ccc}
0 & for & t=1\\
\frac{v_{1}^{t-1}}{v_{i+1}^{t-1}} & for & t>1\end{array}\right.$. 
Then, the strategy $x_i^t$ computed
by strategy iteration on $P(N,m)$ is $p^z$, under the assumption that Lucifer chooses 1 for $\forall t'<t$ and all positions.
\label{lem:matrix(x) corresponds to the matrix in Hoffman-Karp}
\end{lemma}
\begin{proof}
The result follows from Lemma \ref{lem:matrix(x) corresponds to the matrix in Hoffman-Karp, assumption:update} and Lemma \ref{lem:strategy update}.
\end{proof}

\begin{lemma}\label{useful3}
$\forall t,i,j:x_{i,j}^{t}=x_{i,1}^{t}(1-\frac{v_{1}^{t-1}}{v_{i+1}^{t-1}})^{j-1}$, under the assumption that Lucifer chooses 1 for $\forall t'<t$ and all positions.
\end{lemma}
\begin{proof}
This follows from Lemma \ref{fact-useful1}
and Lemma \ref{lem:matrix(x) corresponds to the matrix in Hoffman-Karp}.
\end{proof}

\begin{lemma}\label{lem:AIP}
$\forall t,i>1:x_{i-1,m}^{t}<x_{i,m}^{t}$, under the assumption that Lucifer chooses 1 for $\forall t'<t$ and all positions.
\end{lemma}
\begin{proof}
The result follows from Lemma 
\ref{lem:matrix(x) corresponds to the matrix in Hoffman-Karp}
and Lemma \ref{fact-useful1}, since we have that 
\[
v_{i+1}^{t-1}>v_{i}^{t-1} \Rightarrow
\frac{v_{1}^{t-1}}{v_{i+1}^{t-1}} < \frac{v_{1}^{t-1}}{v_{i}^{t-1}}
\]
from Lemma \ref{lem:positive values} and \ref{lem:increasing values}.
\end{proof}

\begin{lemma}
When applying strategy iteration to $P(N,m)$,
if Lucifers best replies
$y^1,\ldots,y^t$ are all equal to the strategy that chooses 
$1$ in all positions,
then $\forall t'\leq t,i:x_{i,1}^{t'}<x_{i,1}^{t'+1}$
\label{lem:Increasing probabillities}\end{lemma}
\begin{proof}
The proof will be by induction in $t'$.

For $t'=1:$
From Lemma \ref{lem:matrix(x) corresponds to the matrix in Hoffman-Karp},
we have that $x_{i}^{1}$ is the optimal strategy of the row player
of the matrix game $\m 0$.
Since $0<\frac{v_{1}^{0}}{v_{i+1}^{0}}<1$, by Lemma \ref{lem:positive values},
the result follows from Lemma \ref{fact-useful2}.

For $t'>1$, we have the induction hypothesis: 
$\forall i:x_{i,1}^{t'-1}<x_{i,1}^{t'}$.
By Lemma \ref{lem:matrix(x) corresponds to the matrix in Hoffman-Karp}
we have that $x_{i}^{t'}$ is an optimal strategy for the row player in $\m{\frac{v_{1}^{t'-1}}{v_{i+1}^{t'-1}}}$
and $x_{i}^{t'+1}$ is an optimal strategy for the row player in  $\m{\frac{v_{1}^{t'}}{v_{i+1}^{t'}}}$.
By Lemma \ref{new-lemma}, we have 
$\frac{v_{1}^{t'-1}}{v_{i+1}^{t'-1}} = \prod_{j=1}^{i}x_{j,1}^{t'-1}$
and
$\frac{v_{1}^{t'}}{v_{i+1}^{t'}} = \prod_{j=1}^{i}x_{j,1}^{t'}$
From the induction hypothesis and 
Lemma \ref{lem:positive probabillities} we have that $
\prod_{j=1}^{i}x_{j,1}^{t'-1}<\prod_{j=1}^{i}x_{j,1}^{t'}$.

So, $x_{i}^{t'}$ is the optimal strategy for the row player
in $\m{\prod_{j=1}^{i}x_{j,1}^{t'-1}}$ 
and $x_{i}^{t'+1}$ is the optimal strategy for the row player in
$\m{\prod_{j=1}^{i}x_{j,1}^{t'}}$ and the
lemma follows from Lemma \ref{fact-useful2}.\end{proof}

\begin{lemma}
Consider a stationary strategy $x$ for Dante in $P(N,m)$ that
is {\em fully mixed}, i.e., assigns positive probability to
all actions. We may consider $x$ to be a strategy also for $P(k,m)$
for some $k < N$ by identifying each position $i \in \{1,..,k\}$ 
in $P(k,m)$ with
position $i$ in $P(N,m)$.
Suppose a pure strategy $y$ of Lucifer is a best reply to $x$ in $P(N,m)$.
Then, its restriction to positions $1,\ldots,k$ is also a best reply
to $x$ in $P(k,m)$.
\label{lem:Kristoffer argument}\end{lemma}
\begin{proof}

We divide the non-terminal positions of $P(N,m)$ into two sets of positions,
$S = \{1,2,\ldots,k\}$ and $T = \{k+1,\ldots,N\}$.
We note that the only
position the pebble can move to in $T$ directly from $S$ is $k+1$. 
Similarly we note that the only
position the pebble can move to in $S$ directly from $T$ is $1$. 

For a specific fully mixed $x$ and a reply $y$ for $P(N,m)$, an
absorbing Markov process on the set of positions is induced. Let
$Q_{S,T}$ be the probability that the pebble eventually arrives at
position $k+1$, if the process is started in position $1$. Let
$Q_{T,S}$ be the probability that the pebble eventually arrives at $1$,
if the process is started in $k+1$. Let $Q_{S, \mbox{\rm \tiny TRAP}}$
be the probability that the pebble goes to TRAP, without first visiting
$T$, if the process is started in position 1. Similarly, define
$Q_{T,\mbox{\rm \tiny GOAL}}$ to be the probability that the pebble
arrives at GOAL without first visiting 1 if the process is started in
position $k+1$, and $Q_{T,\mbox{\rm \tiny TRAP}}$ to be the
probability that the pebble arrives at TRAP without first visiting 1
if the process is started in position $k+1$.  Observe that $Q_{S,*}$
and $Q_{T,*}$ are probability distributions, since the probability
for a play of infinite length within $S$ and $T$ is 0, because $x$ is
assumed to be fully mixed.

For $u \in \{1,\ldots,k\}$ let $Q_{u,T}$ be the probability that
the pebble reaches $T$ when started in $u$ when $x$ and $y$ are
played.
Note that {\em best} replies $y$ to $x$ in the restricted game $P(k,m)$ 
are characterized by being those $y$ minimizing all 
probabilities $Q_{u,T}$ {\em simultaneously} for all $u \in \{1,\ldots,k\}$,
among all possible $y$, since reaching GOAL in $P(k,m)$ amounts
to reaching $T$ in $P(N,m)$.  But note that in the original $P(N,m)$ game,
the probability of Dante reaching GOAL, when
play starts in some $u \in \{1,\ldots,k\}$ is given by
\begin{equation}\label{min} 
Q_{u,T} Q_{T,\mbox{\rm \tiny GOAL}} \sum_{j=0}^{\infty} (Q_{T,S}Q_{S,T})^j = \frac{Q_{u,T}Q_{T,\mbox{\rm \tiny GOAL}}}
{1 - Q_{T,S}Q_{S,T}}.
\end{equation}
Since $Q_{S,T} = Q_{1,T}$, we have that if the behavior of $y$ in positions
$k+1, \ldots, m$ is fixed (and hence also $Q_{T,*}$ is fixed), 
the behavior of $y$ in positions $1,\ldots,k$
that simultaneously minimizes (\ref{min}) for all $u$ is exactly the
same behavior that simultaneously minimizes $Q_{u,T}$. This concludes the proof.
\end{proof}

\begin{lemma}\label{lem-alwaysone}
When applying strategy iteration to $P(N,m)$, we have that for all
$t \geq 1$, the best reply $y^t$ computed is the one where Lucifer chooses
1 in all positions.
\end{lemma}
\begin{proof}
For $t=1$, we see that for all strategies
Lucifer can select Dante guess correctly with probability $\frac{1}{m}$
as $x^1$ is the uniform choice in each position.
If Lucifer plays 1, Dante will lose the entire game immediately 
with probability $\frac{m-1}{m}$ at each position and advance one
step with probability $\frac{1}{m}$. Any other choice of Lucifer
will preserve the advancement probability but decrease the probability
that Dante loses the game immediately (replacing the probability mass
with a probability of going to the initial position). 
We conclude that choosing 1 is Lucifer's
best reply.

So we only need to look at $t>1$.
We will do the proof using contradiction.

Let $t$ be the lowest value, such that there exists $N,m$ 
so that when applying strategy iteration to $P(N,m)$, 
the reply $y^t$ does not choose 1 in every position.
Also, let $N$ be the lowest such $N$ and let 
$i$ be the smallest $i$ so that $y^t$ does not pick 1 in 
position $i$.

That is, for any position $k<i$, $y^t$ chooses action $1$ in position
$k$, so to determine the best reply $y^t$, we just need to determine
its action in position $i$.  By Lemma \ref{lem:Kristoffer argument},
if we restrict $x^t$ to positions $1,\ldots,i$ and consider the game
$P(i,m)$, the reply $y^t$, restricted to $P(i,m)$, is also a best reply to $x^t$ in this
game. We shall in fact prove that in this game, Lucifer's reply is not
best, unless it chooses $1$, also in position $i$.  This will yield
the desired contradiction. We shall look at each of Lucifer's possible
actions in position $i$.

If Lucifer chooses 1, and play starts in position $i$,
Dante wins $P(i,m)$ if he chooses 1. This Dante
has a probability of $x_{i,1}^{t}$ of doing.

On the other hand, if Lucifer chooses $p>1$, at position $i$, Dante
will go back to state 1 if he chooses $1,\dots,p-1$ and win
immediately if he chooses $p$.

So each time Dante chooses $1,\dots,p-1$, which he does with probability
$\sum_{j=1}^{p-1}x_{i,j}^{t}$ he has to get back to position $i$ from position $1$. Since Lucifer uses strategy $y^t$, Dante needs to chooses $1$ in all positions from $1$ to $i-1$, which he has a probability of $\sum_{j=1}^{p-1}x_{i,j}^{t}$ of doing. Each time
he is at position $i'$ he has a probability of $x_{i,p}^{t}$
to win.

His probability for winning is therefore
\begin{equation}
x_{i,p}^{t}\sum_{l=0}^{\infty}\left(\left(\sum_{j=1}^{p-1}x_{i,j}^{t}\right)\left(\prod_{j=1}^{i-1}x_{j,1}^{t}\right)\right)^{l}  =  \frac{x_{i,p}^{t}}{1-\left(\sum_{j=1}^{p-1}x_{i,j}^{t}\right)\left(\prod_{j=1}^{i-1}x_{j,1}^{t}\right)}
\end{equation}
which, by Lemma \ref{useful3} is equal to
\[ \frac{x_{i,1}^{t}\left(1-\frac{v_{1}^{t-1}}{v_{i+1}^{t-1}}\right)^{p-1}}{1-\left(\sum_{j=1}^{p-1}x_{i,j}^{t}\right)\left(\prod_{j=1}^{i-1}x_{j,1}^{t}\right)}
\]
which, by Lemma \ref{new-lemma} is equal to
\[
\frac{x_{i,1}^{t}\left(1-\prod_{j=1}^{i}x_{j,1}^{t-1}\right)^{p-1}}{1-\left(\sum_{j=1}^{p-1}x_{i,j}^{t}\right)\left(\prod_{j=1}^{i-1}x_{j,1}^{t}\right)}.\]

We will show using induction in $p>1$, that Lucifer is better off
if he always chooses 1, than if he always chooses $p$. That is: 
\begin{equation}\label{whattoprove}
\forall p>1:x_{i,1}^{t}<\frac{x_{i,1}^{t}\left(1-\prod_{j=1}^{i}x_{j,1}^{t-1}\right)^{p-1}}{1-\left(\sum_{j=1}^{p-1}x_{i,j}^{t}\right)\left(\prod_{j=1}^{i-1}x_{j,1}^{t}\right)}.
\end{equation}

For $p=2$, we may argue as follows.
By Lemma \ref{lem:Increasing probabillities} we have that $\prod_{j=1}^{i}x_{j,1}^{t-1}<\prod_{j=1}^{i}x_{j,1}^{t}$. Since $x_{i,1}^t > 0$, this implies
\begin{equation}\label{BC}
x_{i,1}^{t}  <  \frac{x_{i,1}^{t}\left(1-\prod_{j=1}^{i}x_{j,1}^{t-1}\right)}{1-\prod_{j=1}^{i}x_{j,1}^{t}}\end{equation}
which is the statement we wanted to prove.

For $p>2$, we argue as follows. The right hand side of (\ref{whattoprove})
is
\[
\frac{x_{i,1}^{t}\left(1-\prod_{j=1}^{i}x_{j,1}^{t-1}\right)^{p-1}}{1-\left(\sum_{j=1}^{p-1}x_{i,j}^{t}\right)\left(\prod_{j=1}^{i-1}x_{j,1}^{t}\right)}.\]
Applying Lemma \ref{useful3}, this may be rewritten as
 \begin{align}
& \frac{x_{i,1}^{t}\left(1-\prod_{j=1}^{i}x_{j,1}^{t-1}\right)^{p-1}}{1-\left(\sum_{j=1}^{p-1}\left(x_{i,1}^{t}\left(1-\prod_{j=1}^{i}x_{j,1}^{t-1}\right)^{j-1}\right)\right)\left(\prod_{j=1}^{i-1}x_{j,1}^{t}\right)} \nonumber \\
 = &\frac{x_{i,1}^{t}\left(1-\prod_{j=1}^{i}x_{j,1}^{t-1}\right)^{p-1}}{1-\left(\sum_{j=1}^{p-1}\left(1-\prod_{j=1}^{i}x_{j,1}^{t-1}\right)^{j-1}\right)\left(\prod_{j=1}^{i}x_{j,1}^{t}\right)} \nonumber \\
 = & \frac{x_{i,1}^{t}\left(1-\prod_{j=1}^{i}x_{j,1}^{t-1}\right)^{p-1}}{1-\left(\sum_{j=1}^{p-2}\left(1-\prod_{j=1}^{i}x_{j,1}^{t-1}\right)^{j-1}\right)\left(\prod_{j=1}^{i}x_{j,1}^{t}\right)-\left(1-\prod_{j=1}^{i}x_{j,1}^{t-1}\right)^{p-2}\left(\prod_{j=1}^{i}x_{j,1}^{t}\right)} \nonumber \\
 = & \frac{x_{i,1}^{t}\left(1-\prod_{j=1}^{i}x_{j,1}^{t-1}\right)}{\frac{1-\left(\sum_{j=1}^{p-2}\left(1-\prod_{j=1}^{i}x_{j,1}^{t-1}\right)^{j-1}\right)\left(\prod_{j=1}^{i}x_{j,1}^{t}\right)}{\left(1-\prod_{j=1}^{i}x_{j,1}^{t-1}\right)^{p-2}}-\prod_{j=1}^{i}x_{j,1}^{t}}\label{last}
\end{align}
To bound (\ref{last}), we use the induction hypothesis:
\begin{equation*}
x_{i,1}^{t}  <  \frac{x_{i,1}^{t}\left(1-\prod_{j=1}^{i}x_{j,1}^{t-1}\right)^{p-2}}{1-\left(\sum_{j=1}^{p-2}x_{i,j}^{t}\right)\left(\prod_{j=1}^{i-1}x_{j,1}^{t}\right)}
\end{equation*}
We note that the induction hypothesis implies
\begin{equation*}
1  >  \frac{1-\left(\sum_{j=1}^{p-2}x_{i,j}^{t}\right)\left(\prod_{j=1}^{i-1}x_{j,1}^{t}\right)}{\left(1-\prod_{j=1}^{i}x_{j,1}^{t-1}\right)^{p-2}}
\end{equation*}
and conclude that the expression (\ref{last}) is at least: 
\[ \frac{x_{i,1}^{t}\left(1-\prod_{j=1}^{i}x_{j,1}^{t-1}\right)}{1-\prod_{j=1}^{i}x_{j,1}^{t}} \]
which, by equation (\ref{BC}) is strictly greater than $x_{i,1}^t$, as desired.
\end{proof}


\begin{lemma}
\label{onemany}
Let $x^t$ be the behavior strategies computed when applying strategy
iteration to $P(N,m)$.
Let $\hat x^t$ be the behavior strategies computed when applying strategy
iteration to $P(1,m)$.
Then, for all $t$, $x_1^t = \hat x_1^t$.
\end{lemma}
\begin{proof}
We show this by induction in $t$. 
For $t=1$, both $x_1^1 = \frac{1}{m}$ and $\hat x_1^t=\frac{1}{m}$.
For $t>1$, Lemma \ref{lem-alwaysone} states that $y^t$ chooses 1
in every position.
By Lemma \ref{lem:matrix(x) corresponds to the matrix in Hoffman-Karp},
we have that
$x_1^t = p^z$, where $z = \frac{v_{1}^{t-1}}{v_{2}^{t-1}} = x_1^{t-1}$,
where the last equation is by Lemma \ref{new-lemma}.

On the other hand, applying strategy iteration to $P(1,m)$, yielding
strategies $\hat x^t_1$,
we similarly get $\hat x_1^t = p^z$, where $z = \hat x_1^{t-1}$. Since
$x_1^{t-1} = \hat x_1^{t-1}$ by induction, we are done.
\end{proof}

\begin{lemma}\label{lem:SV} Applying strategy iteration to $P(1,m)$ yields
valuations $v^t = \tilde v^t$, i.e. strategy iteration computes the same valuations as value iteration.
\end{lemma}
\begin{proof}
We show this by induction in $t$. By Lemma \ref{lem-alwaysone}, $y^{t-1}$ is the
strategy that chooses $1$. Thus, Dante wins if and only if he chooses
1 in the first round and we have $v_1^{t-1} = x_{1,1}^{t-1}$.
On the other hand, by Lemma 
\ref{lem:matrix(x) corresponds to the matrix in Hoffman-Karp},
we have that $x_{1,1}^t = p^z$ where $z = x_{1,1}^{t-1}$.
Thus, $v^t = p^z$ where $z = v_{1}^{t-1}$. Inspecting the
value iteration algorithm we find that we also have
that $\tilde v^t = p^z$ where $z = \tilde v_{1}^{t-1}$, and
since we can see by inspection that we also have $v^1 = \tilde v^1$, we
are done.
\end{proof}
Note that Lemma \ref{lem:SV} together with Corollary 
\ref{cor-valitonepos} yields our previously stated claim that
strategy iteration may need exponential time to achieve non-trivial
approximations for a one-position game.

Finally, the proof of Lemma \ref{lem-patsi}, 
(stating that when applying strategy iteration to $P(N,m)$, the patience 
of the strategy $x^t$ computed in iteration $t$ is 
at most $emt$)
\noindent
\begin{proof}[Proof of Lemma \ref{lem-patsi}]
By Lemma \ref{lem:AIP}, Lemma \ref{fact-useful1} and
Lemma \ref{fact-useful2}, 
we have that the smallest behavior probability in $x^t$ is
$x^t_{1,m}$, i.e., the probability of playing $m$ in the start
position where Dante still has to guess correctly $N$ times
to win.

Then, by Lemma \ref{onemany}, to estimate this probability, we can
consider $P(1,m)$ instead of $P(N,m)$. In fact we shall consider
the valuations $v^t$ computed when applying strategy iteration
to $P(1,m)$. By Lemma \ref{lem:SV} the values computed are the
same as those $\tilde v^t$ computed by value iteration on $P(1,m)$. 
So, by Theorem \ref{thm-valitonepos} and
Lemma \ref{lem-valit} we have that
$v^t \leq 1-(1-\frac{1}{m})(\frac{1}{mT})^{1/(m-1)}$.
That is, $1-v^t \geq (1-\frac{1}{m})(\frac{1}{mt})^{1/(m-1)}$.
Now, Lemma \ref{fact-useful1} tells us that 
$x^t_{1,m} \geq ((1-\frac{1}{m})(\frac{1}{mt})^{1/(m-1)})^{m-1} =
(1-\frac{1}{m})^{m-1} (\frac{1}{mt}) \geq \frac{1}{emt}$ and
we are done.
\end{proof}

\section*{Acknowledgements}
First and foremost, we would like to thank Uri Zwick for extremely
helpful discussions and Kousha Etessami for being instrumental
for starting this research. We would also like to thank Vladimir
V. Podolskii for helpful discussions. A preliminary version of this paper \cite{HIMCSR} appeared in the proceeings of CSR'11.

%

\bibliographystyle{plain}      


%
%

\end{document}